\documentclass[reqno,11pt]{amsart}

\usepackage{amsmath,amsthm,amscd,amsfonts,amssymb}
\usepackage{latexsym}
\usepackage{mathabx}

\newcommand{\labelnummer}{\mbox{\normalfont (\roman{numcount})}}%

\makeatletter

  {\let\curlabelspeicher\@currentlabel%
    \begin{list}{\labelnummer}%
      {\usecounter{numcount}\leftmargin0pt%
        \topsep0.5ex\partopsep2ex\parsep0pt\itemsep0ex\@plus1\p@%
        \labelwidth2.5em\itemindent3.5em\labelsep1em%
      }%
    \let\saveitem\item%
    \def\item{\saveitem%
      \def\@currentlabel{{\upshape\curlabelspeicher}$\,$\labelnummer}}%
    \let\savelabel\label%
    \def\label##1{\savelabel{##1}%
      \@bsphack%
        \ifmmode\else%
          \protected@write\@auxout{}%
          {\string\newlabel{##1item}{{\labelnummer}{\thepage}}}%
        \fi%
      \@esphack%
    }%
  }{\end{list}}%

\usepackage[usenames]{color}
\usepackage{comment}
\newcommand{\CC}{\mathbb C}

\newcommand{\NN}{\mathbb N}

\newcommand{\RR}{\mathbb R}
\newcommand{\TT}{\mathbb T}
\newcommand{\ZZ}{\mathbb Z}

\newcommand{\mm}{\mathcal M}

\newcommand{\ee}{\mathcal E}

\newcommand{\Ss}{\mathcal S}

\newtheorem{thm}{Theorem}[section]
\newtheorem{lemma}[thm]{Lemma}
\newtheorem{cor}[thm]{Corollary}
\newtheorem{prop}[thm]{Proposition}
\newtheorem{definition}[thm]{Definition}
\newtheorem{rem}[thm]{Remark}

\newtheorem{assum}[thm]{Assumption}

\def\half{\frac{1}{2}}
\newcommand{\ov}[1]{\frac{1}{#1}}

\newcommand{\dist}{\rm{dist}}

\newenvironment{proofn}[1][]{\noindent\textbf{Proof}#1\textbf{:} }{\ \hfill \rule{0.5em}{0.5em}\\[2mm]}

\renewcommand{\i}{\mathrm{i}}
\newcommand{\e}{\mathrm{e}}

\newcommand{\bel}{\begin{equation} \label}
\newcommand{\eeq}{\end{equation}}
\newcommand{\beq}{\begin{equation}}

\newcommand{\ba}{\begin{array}}
\newcommand{\ea}{\end{array}}
\newcommand{\bea}{\begin{eqnarray}}
\newcommand{\eea}{\end{eqnarray}}

\newcommand{\SCHR}{SCHR\"ODINGER }
\newcommand{\Schr}{Schr\"odinger }

\begin{document}

%
%
%
%

\title[Eigenfunctions and quantum transport]{Eigenfunctions and Quantum Transport with Applications to Trimmed \SCHR Operators}

\author[P.\ D.\ Hislop]{Peter D.\ Hislop}
\address{Department of Mathematics,
    University of Kentucky,
    Lexington, Kentucky  40506-0027, USA}
\email{peter.hislop@uky.edu}

\author[W.\ Kirsch]{Werner Kirsch}
\address{Fakult\"{a}t f\"ur Mathematik und Informatik,  FernUniversit\"at in Hagen, 58097 Hagen, Germany}
\email{werner.kirsch@fernuni-hagen.de} 

\author[M.\ Krishna]{M.\ Krishna}
\address{Ashoka University, Plot 2, Rajiv Gandhi Education City, Rai, Haryana 
131029 India}
\email{krishna.maddaly@ashoka.edu.in}

\begin{center}
\emph{Dedicated to our friend Abel Klein}
\end{center}

\begin{abstract}
We provide a simple proof of dynamical delocalization, that is, time-increasing lower bounds on quantum transport for discrete, one-particle \Schr operators on $\ell^2 (\ZZ^d)$, provided solutions to the \Schr equation satisfy certain growth conditions. The proof is based on basic resolvent identities and the Combes-Thomas estimate on the exponential decay of the Green's function. As a consequence, we prove that generalized eigenfunctions for energies outside the spectrum of $H$ must grow exponentially in some directions. We also prove that if $H$ has any absolutely continuous spectrum, then the \Schr operator exhibits dynamical delocalization. We apply the general result to $\Gamma$-trimmed \Schr operators, with periodic $\Gamma$, and prove dynamical delocalization for these operators. These results also apply to the $\Gamma$-trimmed Anderson model, providing a random, ergodic model exhibiting both dynamical localization in an energy interval and dynamical delocalization.  
\end{abstract}

\medskip

\thanks{PDH is partially supported by Simons Foundation Collaboration Grant for Mathematicians No.\ 843327.}

\maketitle \thispagestyle{empty}

\tableofcontents

\section{Introduction: multi-dimensional quantum transport}

Many papers are devoted to exploring the relationship between the spectrum of a self-adjoint \Schr operator $H_V$ on the Hilbert space $\ell^2(\ZZ^d)$ or $L^2 (\RR^d)$,  the properties of solutions to the \Schr equation $i\partial_t \varphi_t = H_V \varphi_t$,  the growth rate of the generalized eigenfunctions $\psi_E$ solving $H_V \psi_E = E \psi_E$, and the transport of the model as measured by moments of the position operator $X$. Following \cite{JSS, GKS}, we will say that a \Schr operator exhibits \emph{dynamical delocalization} if some moment of the position operator increases without bound in time. We prove a simple condition on generalized eigenfunctions of \Schr operators that implies dynamical delocalization and apply it to certain $\Gamma$-trimmed models in $\ell^2(\ZZ^d)$, $d \geq 1$.

To explain this, we denote by $\varphi_t $ the solution to the \Schr equation with initial condition $\varphi_0$:
$$
i \partial_t \varphi_t = H_V \varphi_t, ~~~~ \varphi_{t=0} = \varphi_0 \in L^2 (\RR^d) ~~{\rm or} ~~ \ell^2(\ZZ^d) .$$
This may be written as  $\varphi_t := U_{H_V}(t) \varphi_0$,
where $U_{H_V}(t) := e^{-i t H_V}$, $t \in \RR$, is the unitary evolution group generated by $H_V$. 

Quantum transport is often measured by the mean square displacement (MSD) $\mm_{2}(T)$ of a quantum wave packet. 
For lattice models on $\ell^2 (\ZZ^d)$, we take the initial condition $\varphi_0 (n) = \delta_{0n}$, the Kronecker delta function, and define the averaged $q^{th}$-moment of the position operator $X$ to be:
\begin{align}\label{defmm0}
   \mm_{q}(T)~=~\ov{T}\int_{0}^{\infty} e^{-\frac{t}{T}}\,\langle \delta_0,
   e^{iH_Vt}|X|^q e^{-iH_Vt}\, \delta_0\rangle\;dt  .
\end{align}
If $\mm_{2}(T) \sim T^\alpha$, we say that the transport is localized, if $\alpha =0$, the transport is diffusive if $\alpha = \frac{1}{2}$, and it is ballistic if $\alpha = 1$. 


In the context of random \Schr operators, Anderson localization (pure point spectrum with exponentially decaying eigenfunctions) can often be strengthened to dynamical localization. Dynamical localization implies that the moments $\mathcal{M}_q(T)$ are bounded for all $q \geq 1$, see, for example, \cite{GK-character}. On the other hand, if $\mathcal{M}_q(T)$ is bounded from below by an increasing function of $T$ for some $q$, one speaks of dynamical delocalization, see, for example, \cite{JSS} and \cite{GKS}.  The one-dimensional, discrete, random polymer model is studied in \cite{JSS}. The authors prove that at critical energies, the quantum transport is superdiffusive almost surely. In \cite{GKS}, the authors prove that for the random Landau Hamiltonian, the local transport exponent at an energy near each Landau level is strictly positive. Hence, both of these models exhibit dynamical delocalization. . 

A simplified version of one of our main results, Theorem \ref{thm:mainres}, states that if there is a subset $J \subset \RR$ of positive Lebesgue measure so that for almost every $E \in J$, there is a generalized eigenfunction $\psi_E$, and a finite constant $C_E > 0$, so $| \psi_E (x) | < C_E$, for a.e.\ all $E \in J$, then the $q^{th}$-moment of the position operator $X$ \eqref{defmm0} satisfies
\beq\label{eq:main1}
\mathcal{M}_q(T) \geq C_J T ,
\eeq
provided $q > d$. 
In general, Theorem \ref{thm:mainres} relates the growth rates of generalized eigenfunctions of $H_V$ to lower bounds on the moments of the position operator establishing dynamical delocalization.. 

\subsection{A few facts on quantum transport}


We mention some results on transport for multi-dimensional \Schr operators. 
For arbitrary dimension $d \geq 2$, it seems that first lower bounds on the moments of the position operators are due to Guarneri \cite{guarneri}. Combes \cite{combes,bcm} simplified Guarnari's proof and extended these results to \Schr operators on $\RR^d$, using the Strichartz inequality \cite{strichartz}. These original works measured quantum transport in terms of the upper and lower transport exponents. These are defined in terms of the spectral measure $\mu_{\varphi_0}^H$ associated with $H :=H_V$ and the initial state $\varphi_0$. The lower transport exponent is defined by
\beq\label{eq:sp_mLower}
\beta_- (x):=  \liminf_{\epsilon \rightarrow 0^+} \frac{\log \{\mu_{\varphi_0}^H (B_\epsilon (x))\}}{\log \epsilon} ,
\eeq
and the upper transport exponent is defined by 
\beq\label{eq:sp_mUpper}
\beta_+ (x):=  \limsup_{\epsilon \rightarrow 0^+} \frac{\log \{\mu_{\varphi_0}^H (B_\epsilon (x))\}}{\log \epsilon} .
\eeq
We note that the transport exponents satisfy $0 \leq \beta_\pm(x) \leq d$. 

The Combes-Guarnari \cite{combes,bcm} results state that for initial state $\varphi_0 = \delta_0$, and for any $q > 0$ and any $\nu >0$, there is a finite constant $C(q,d,\nu,\varphi_0) >0$, so that 
\beq
\mathcal{M}_q (T) > C(q,d,\nu,\varphi_0) T^{\frac{q}{d}(\gamma_0 - \nu)} ,
\eeq
where $\gamma_0$ is defined to be the $ \mu_{\varphi_0}^H$ essential supremum of 
the lower transport exponent $\beta_- (x)$. In \cite{bcm}, it is proved that $\gamma_0$ is the Hausdorff dimension of the spectral measure $ \mu_{\varphi_0}^H$. 



Other results on transport of multi-dimensional \Schr are due to Christ, Kiselev and Last \cite{CKL} and by Kiselev and Last \cite{KL}. Their results are similar to ours. In fact, \cite[Theorem 1.2]{KL} gives a sufficient condition in terms of generalized eigenfunctions insuring quantum transport:

\begin{thm}\cite[Theorem 1.2]{KL} \label{thm:kl1}
Let $\psi$ be a vector for which there exists a Borel set $S \subset \RR$ of positive $\mu_H^\psi$ measure, such that the restriction of $\mu_H^\psi$ to $S$ is $\alpha$-continuous and, in
addition, the real generalized eigenfunctions $u(x,E)$ for all $E \in S$ satisfy
\beq\label{eq:kl1}
\limsup_{R \rightarrow \infty} \frac{1}{R^\gamma} \int_{B_R(0)} u(x,E)^2
  < \infty,
\eeq
for some $\gamma$ such that $0 < \gamma  < d$. Then, for any $m > 0$, there exists a constant $C_m >0$, such
that
\beq\label{eq:kl2}
\mathcal{M}_m(T)  \geq C_m T^{\frac{m \alpha}{\gamma}} ,
\eeq
for all $T > 0 $.
\end{thm}

We note that this result requires conditions on \textbf{both} the generalized eigenfunctions and on the local spectral measure, whereas our main result, Theorem \ref{thm:mainres}, depends only on a growth condition on the generalized eigenfunctions. This allows us to prove that certain models exhibit some nontrivial transport, although we do not obtain a precise transport exponent. 
We comment further on the relation of this theorem to our results for $\Gamma$-trimmed \Schr operators in Remark \ref{rmk:kl1}.

\subsection{A few facts on $\Gamma$-trimmed \Schr operators}\label{subsec:trim_intro1}

A $\Gamma$-trimmed \Schr operator is one for which the potential is supported on a sublattice $\Gamma \subset \ZZ^d$. We are interested in the case when $\Gamma$ is a periodic sublattice. The case $\Gamma=\ZZ^{d} $ corresponds to the usual \Schr operator. If $\Gamma\not=\ZZ^{d}$, we have `missing sites', that is, a periodic lattice of sites without any potential.
The $\Gamma$-trimmed Anderson model is a version of the usual Anderson model for which the random potential is supported on a sublattice $\Gamma \subset \ZZ^d$.
In this case, $\Gamma=\ZZ^{d} $ corresponds to the classical Anderson model. Such models exhibit localization: intervals of dense pure point spectrum with exponential decaying eigenvectors for almost every realization of the random potential.  
It is of interest to explore if localization persists for $\Gamma$-trimmed Anderson models. For example, if the connected components of $\Gamma^{c}=\ZZ^{d}\setminus\Gamma $ are finite, the $\Gamma$-trimmed Anderson model has pure point spectrum for low energy or for high disorder (see e.\,g. \cite{JO} in connection with \cite{K94}, \cite{CR} or \cite{EK}; for further studies on such models see \cite{KiKr1,ES,KiKr2}).

More generally, if $\Gamma \subset \ZZ^d$ is a nonempty, periodic, subset of $\ZZ^d$, 
the corresponding $\Gamma$-trimmed Anderson model with deterministic spectrum $\Sigma$ has pure point spectrum for energies in $\Sigma$ outside of the spectrum of $H_{\Gamma^c}$, the restriction of $H$ to $\Gamma^c$ \cite[Theorem 2]{ES} at least for large disorder and some additional assumptions on $\Gamma$. The non-emptiness of this energy interval $\Sigma \backslash H_{\Gamma^c}$ was proved in \cite{EK}. 
The crucial ingredient for the localization proofs mentioned above is the validity of a Wegner estimate. Such an estimate is \emph{not} valid for \emph{all} energies if the set $\Gamma $ is too sparse, for example, if it consists of isolated hyperplanes (see \cite{ES} and \cite{KiKr2}). In particular, one can show that the Green's function does not display the `typical' exponential decay which is at the heart of all known proofs of Anderson localiaztion (see \cite{ES} and \cite{KiKr2}). This could be interpreted as a breakdown of exponential localization.

We apply the results of sections \ref{sec:definitions} and \ref{sec:resolvents} to $\Gamma $-trimmed \Schr operators. We first make \emph{no} assumptions on the potentials on $\Gamma$. In particular, we don't assume that the potential is random or even nonzero (although this case is well known). We will prove that $\Gamma$-trimmed \Schr operators, for nontrivial families of sublattices $\Gamma \subset \ZZ^d$, such that $\Gamma$ is periodic, exhibit nontrivial transport, a signature of dynamical delocalization. We then apply our results to the $\Gamma$-trimmed Anderson model and prove there exists nontrivial transport. This result extends those of \cite{KiKr2}. In that paper, the authors proved that for a quantum waveguide version of the $\Gamma$-trimmed Anderson model, there is dense pure point spectrum outside the spectrum of  $H_{\Gamma^c}$ and some absolutely continuous spectrum almost surely. In this paper, we prove that the more general $\Gamma$-trimmed Anderson model always exhibits nontrivial quantum transport, a result which, of course, applies to the quantum waveguide models of \cite{KiKr2}. This general result on dynamical delocalization for the $\Gamma$-trimmed Anderson model was conjectured in \cite{ES}.


\section{Lattice \Schr operators and quantum transport: Definitions and results}
\label{sec:definitions}

In this section, we define the model on $\ell^2(\ZZ^d)$ and the mean moments of the position operator $X$ on $\ell^2(\ZZ^d)$ that will be used to measure quantum transport. We give a summary of the main results. 


\subsection{Definitions and model}\label{subsec:model1}

We consider discrete Schr\"{o}dinger  operators $H$ on the Hilbert space $\ell^{2}(\ZZ^{d})$. Unless stated otherwise, we use the norm $|n|=\max_{1\le\nu\leq d} |n_{\nu}| $ on $\ZZ^{d} $. 
A cube of side length $2L+1$  in $\ZZ^d$ centered at $n_0 \in \ZZ^d$ is denoted by $\Lambda_L(n_0)$, and its boundary is denoted by $\partial\Lambda_L(n_0)$:  
\begin{align}
  &\Lambda_{L}(n_{0})~:=~\{ n\in\ZZ^{d}\mid | n-n_{0} | \leq L\}, \text{ a cube of side length $L$}\,,   \nonumber \\
&\partial\Lambda_{L}(n_{0}):=\{ n\in\ZZ^{d}\mid | n-n_{0} |= L\}, \text{ the boundary of } \Lambda_{L}(n_{0})  \nonumber\\
&\overline{\partial}\Lambda_{L}(n_0) := \partial\Lambda_{L}(n_{0})\cup\partial\Lambda_{L+1}(n_{0}),\quad\text{ the `enlarged' boundary}\,\nonumber.
\end{align}
If $n_0=0$, we set $\Lambda_{L} :=\Lambda_{L}(0)$, and similarly for the boundary set. 

We define the discrete Laplacian on $ \ell^{2}(\ZZ^{d})$ by
\begin{align}
   H_{0}u(n)~=~\sum_{|j|=1}u(n+j)\, ,
\end{align}
and note that the spectrum of $H_0$, denoted $\sigma(H_0)$, is $\sigma(H_0) = [-2d, 2d]$.
For any bounded function $V:\ZZ^{d}\to\RR$,
the \Schr operator $H_V$ is defined by
\begin{align}\label{defSchr}
   H_V ~ :=~H_{0}+V .
\end{align}
The bounded operator $H_V$ is self-adjoint on $\ell^2 ( \ZZ^d )$.
When $V$ does not play any special role, we denote $H_V$ simply by $H$. 
For any cube $\Lambda_L(n_0)$, we write $H_{\Lambda_L(n_0)}$ for the restriction of $H_0$ to $\Lambda_L(n_0)$ with simple boundary conditions. 

\subsection{Properties of solutions to the \Schr equation}\label{subsec:model2}

It is well known that the spectrum of $H$ is intimately connected with the existence of solutions to the \Schr equation. 

\begin{definition}
   A function $u:\ZZ^{d}\to\CC $  which satisfies the difference eigenvalue equation
   \begin{align}\label{eq:ev_eq1}
      Hu(n)~=~E\,u(n) , \qquad \text{for all } n\in\ZZ^{d} ,
   \end{align}
   for some $E\in\CC$ is called a \emph{generalized eigenfunction}.
   \begin{enumerate}
\item A generalized eigenfunction is called a \emph{polynomially-bounded generalized eigenfunction} (pb-generalized eigenfunction) if there are constants $k\in\NN$ and $C>0$ such that
   \begin{align}\label{genE}
      |u(n) |~\leq~C\,\langle n \rangle^{k} \qquad\text{for all } n\in\ZZ^{d}\,,
   \end{align}
   where $\langle n \rangle=(1+| n |^{2})^{\half} $. 
In this case, 
the number $E$ is called a \emph{pb-generalized eigenvalue}. 

\medskip
\item Let $\ee(H)$ denote the set of all pb-generalized eigenvalues of $H$, and let 
$\ee_{0}(H) \subset \ee(H)$ be the set of $E\in\ee(H)$ such that for any $k>\frac{d}{2}$, there is a pb-generalized eigenfunction
$\psi_E $ and a constant $C$ such that
\begin{align}\label{eq:pb_gen_ef1}
   |\psi_E(n) |~\leq~C\,\langle n \rangle^{k} \qquad\text{for all } n\in\ZZ^{d}\,.
\end{align}

\medskip
\item  A solution of \eqref{eq:ev_eq1} for $E \in \RR$ is called an \emph{eigenfunction} of $H$ if $u \in \ell^2 ( \ZZ^d)$. In this case, the energy $E$ is simply called an eigenvalue of $H$. 
\end{enumerate}

\end{definition}

%
It is well-known that the pb-generalized eigenfunctions play a key role in determining the spectrum of $H$. 

\begin{thm}[Sch'nol]\label{thm:Sch} Suppose $H$ is a self-adjoint Schr\"{o}dinger operator as in \eqref{defSchr}.
\begin{enumerate}
\item Recalling that $\ee(H)$ denotes the set of pb-generalized eigenvalues of $H$
we have 
\begin{enumerate}
   \item $\overline{\ee(H)}~=~\sigma(H) $.
   \item If $\rho $ is a spectral measure of $H$, then $\rho\Big(\sigma(H)\setminus\ee(H)\Big)=0$.
\end{enumerate}
   
   \item  Recalling that $\ee_0(H)$ denotes the subset with pb-generalized eigenvalues with generalized eigenfunctions satisfying \eqref{eq:pb_gen_ef1}, we  have
\begin{enumerate}
   \item $\overline{\ee_0(H)}~=~\sigma(H) $.
   \item If $\rho $ is a spectral measure of $H$, then $\rho\Big(\sigma(H)\setminus\ee_0(H)\Big)=0$.
\end{enumerate}

\end{enumerate}

\end{thm}

Part 1 of this theorem goes back to \cite{Schnol}, see also \cite[section 2.4]{BS:semigroups}. A proof of part 2 can be found, for example, in \cite[Theorem 7.1, section 7.1] {Invitation}, see also \cite{BS:semigroups}.


It follows that for $E\not\in\sigma(H)$, there cannot be a pb-generalized eigenfunction of the Schr\"{o}dinger operator.
As a consequence of our main result, we prove that any solution of $H\psi_E=E\psi_E$, for $E\not\in\sigma(H)$, must grow exponentially
in a certain sense made precise in Theorem \ref{thm:eigenoutside} and Corollary \ref{cor:eigenoutside}.

\bigskip

\subsection{Measure of quantum transport}\label{subsec:q_transport1}

We measure the growth of generalized eigenfunctions using a nonnegative function $\varphi$. Let $\varphi:\ZZ^{d}\to (0,\infty) $, be any function with a subexponential upper bound:
\begin{align}\label{assphi}
\varphi(n)\leq C e^{| n |^{\beta}} \qquad \text{for some } \beta<1 \,.
\end{align}
Denote by $\varphi(X) $ the operator of multiplication by $\varphi $, i.\,e.
\begin{align}
   \varphi(X)u(n)~=~\varphi(n)u(n)\,.
\end{align}

We are interested in the large $T$ behavior of the transport function $\mm_{\varphi;k}(T)$ defined, for any $k \in \ZZ^d$, by 
\begin{align}\label{defmm}
   \mm_{\varphi;k}(T)~ :=~\ov{T}\int_{0}^{\infty} e^{-\frac{t}{T}}\,\langle \delta_k,e^{iHt}\varphi(X) e^{-iHt}\, \delta_k\rangle\;dt ,
\end{align}
where $\delta_k(m) = 1$, if $m=k$, and zero otherwise. We note that if $H$ commutes with the unitary representation of an additive subgroup $\Gamma \subset \ZZ^d$, with $k \in \Gamma$, then $\mm_{\varphi;k}(T) = \mm_{\varphi_k;0}(T)$, where $\varphi_k(m) = \varphi(m-k)$. 
We will relate the behavior of $\mm_{\varphi;k}(T)$ for large $T$ to the behavior of generalized eigenfunctions near infinity.

An alternative quantity to measure quantum transport is the Cesaro mean
\begin{align}\label{defCM}
   M_{\varphi;k}(T)~=~\ov{T}\int_{0}^{T} \langle \delta_k,e^{iHt}\varphi(X) e^{-iHt}\,\delta_k\rangle\;dt\,.
\end{align}
These two quantities are closely related.

\begin{prop}\label{prop:Mphi} For any growth function $\varphi$ as in \eqref{assphi},  we have
\begin{enumerate}
   \item $M_{\varphi;k}(T)~\leq~ \ov{\e}\,\mm_{\varphi;k}(T)$,
   \item For any $\varepsilon>0 $ there is a constant $C$ such that $\mm_{\varphi;k}(T)~\leq~ C\,T^{\varepsilon}\;M_{\varphi;k}(T) $.
\end{enumerate}
\end{prop}
The proof of Proposition \ref{prop:Mphi} follows along the lines of the Lemma in \cite[section 3]{GuarneriSB}.

\medskip

Our main result is the following theorem relating the existence of generalized eigenfunctions to quantum transport. 

%

\begin{thm}\label{thm:main2}
 Let $I\subset\RR$ is a set of positive Lebesgue measure, and let $\varphi$ be some growth function satisfying \eqref{assphi}.
 Suppose that there exists $\nu := \nu_I \in [0,1)$, and a point $n_0 \in \ZZ^d$, so that  for Lebesgue almost every $E \in I$, there is a pb-generalized eigenfunction $\psi_E$ of $H$, $H \psi_E = E \psi_E$,  with $\psi_E(n_0) \neq 0$, and a finite constant $A_E > 0$,  so that $\psi_E$ satisfies the growth condition
   \begin{align}\label{phibehav}
      \sum_{|n|\leq \Lambda_L(n_0)}\frac{|\psi_{E}(n)|^{2}}{\varphi(n)}~\leq~A_E \,L^{\nu}\,,
   \end{align}
   for all $L>0$. Then,  for any $\alpha>1$, there is a finite constant $C := C(d,\nu,I) >0$ with
   \begin{align}
      \mm_{\varphi;n_0}(T)~\geq~C\,T^{1-\alpha\nu}\,.
   \end{align}
\end{thm}

\medskip

Our main example of $\varphi (x)$ is the function $\langle x \rangle^{q}$ corresponding to the $q^{th}$-moment of the position operator. We therefore simplify notation by writing
$\mm_{q;k}(T) :=\mm_{\langle x \rangle^{q};k}(T) $, for any $k \in \ZZ^d$.
Rephrasing Theorem \ref{thm:main2} with growth function $\langle x \rangle^{q}$, we have the corollary: 


\begin{cor}\label{cor:bnd}
   Suppose the generalized eigenfunctions $\psi_{E} $ (as in Theorem \ref{thm:main2}) are bounded, then for $q>d$
   \begin{align}
      \mm_{q;k}(T)~\geq~C\,T\,.
   \end{align}
   If $d-1<q\leq d$ then
   \begin{align}
      \mm_{q;k}(T)~\geq~C\,T^{r}\,,
   \end{align}
   for any $r<1+q-d $.
\end{cor}

\medskip
We also obtain another corollary. The existence of absolutely continuous spectrum for a discrete, bounded, self-adjoint \Schr operator implies nontrivial transport.

\begin{cor}\label{cor:ac_sp1}
   If the spectral measure of $H$ has an absolutely continuous component, then, for some $n_0 \in \ZZ^d$, 
   \begin{align}
      \mm_{q;n_0}(T)~\geq~C\,T\,,
   \end{align}
   if $q>2d $.
\end{cor}

\begin{proof}(Corollary \ref{cor:ac_sp1})
If $\sigma_{ac} (H) \neq \emptyset$, then there exists a set $I \subset  \sigma_{ac} (H)$ of positive Lebesgue measure. For Lebesgue almost every $E \in I$, there exists a pb-generalized eigenfunction with exponent $k= \frac{d}{2} + \epsilon$, for any $\epsilon > 0$, as stated in part 2 of Theorem \ref{thm:Sch}. The result then follows from Lemma \ref{lemma:non_zero1} and Theorem \ref{thm:main2}.
\end{proof}

This result does not seem to have appeared in the literature. 
The proof of Theorem \ref{thm:main2} appears in sections \ref{sec:resolvents} and \ref{sec:green1}.

\medskip

Our main application of Theorem \ref{thm:main2}, which was also the motivation for this work, is to the family of $\Gamma$-trimmed \Schr operators. We consider a periodic subset $\Gamma \subset \ZZ^d$. The potential $V$ is supported on $\Gamma$. The cases of interest is when $\Gamma$ is a $d_1$-dimensional lattice in $\ZZ^d$ and $d_2 := d - d_1 \geq 2$. These are sparse potentials and one is interested in the stability of dynamical delocalization under various choices of $\Gamma$.  
The details of the families of $\Gamma$ for which we can prove dynamical delocalization are given in section \ref{sec:trimmed}. 


\begin{thm}\label{thm:main_trimmed1}
A $\Gamma$-trimmed \Schr operator, as defined in section \ref{sec:trimmed}, with $d_2 \geq 2$. Then, the \Schr operator exhibits dynamical delocalization. This nontrival quantum transport is manifest in the lower bound  
\beq\label{eq:trim1}
\mm_{q, n_0}(T) \geq C T, ~~~~\forall q > d_1 + 1,
\eeq
for some $n_0 \in \Gamma^c$, and where $d_1$ is defined in Assumption \ref{assupm:trim2}.
\end{thm}

$\Gamma$-trimmed Anderson models have been recently studied and various authors have proved intervals in which they exhibit dynamical localization almost surely, see \cite{EK,ES,CR}. We prove that for certain $\Gamma$-trimmed Anderson models, there is also dynamical delocalization.



\section{Transport and Resolvents}\label{sec:resolvents}

We shall prove Theorem \ref{thm:main2} via estimates on resolvents $R_H(z) := (H-z)^{-1}$, for $z \not\in \sigma(H)$,  presented in this and the following section \ref{sec:green1}. We conclude this section with some estimates on generalized eigenfunctions. 
The proof of Theorem \ref{thm:main2} is presented at the end of section \ref{sec:green1}. We begin with the following theorem that will allow us to convert estimates on matrix elements of time evolved states to estimates on matrix elements of the resolvent as in Theorem \ref{thm:main2}.

\begin{thm}\label{thm:timeres} For $T>0$ and $\varphi $ as in \eqref{assphi} we have
\begin{align}\label{eq:mm0}
   \mm_{\varphi;n_0}(T)~=~\frac{1}{2\pi T}\,\int_{-\infty}^{\infty}\,\sum_{n\in\ZZ^{d}}\,\varphi(n)\,\Big|\big(H-E-\frac{1}{2T}i\big)^{-1}(n_0,n)\Big|^{2}\,dE
\end{align}
\end{thm}
This result is cited and used in \cite[section 3]{JSB} and in \cite[Lemma 6.3]{GK-character}. We present a proof based similar to the one in \cite[Lemma 6.3]{GK-character}. 

\begin{proof}
We begin with an identity for any $T > 0$ and any self-adjoint operator $H$: 
\beq\label{eq:resolv1}
-i R_H(E + i \frac{1}{2T} ) = \int_0^\infty e^{-i t( H - i\frac{1}{2T} - E)} ~dt.
\eeq
The integral is absolutely convergent. It follows that 
\beq\label{eq:resolv2}
\langle \delta_n, R_H( E + i \frac{1}{2T}) \delta_{n_0} \rangle = i \int_0^\infty e^{itE} 
\langle \delta_n, e^{-it( H - i \frac{1}{2T})} \delta_{n_0} \rangle ~dt .
\eeq
Extending the $L^2$-function $g(t) := \langle \delta_n, e^{-it( H - i \frac{1}{2T})} \delta_{n_0} \rangle$ to be zero for $t < 0$, Plancherel's Theorem gives
\beq\label{eq:planch1}
\frac{1}{2 \pi} \int_{-\infty}^\infty |\langle \delta_n, R_H( E + i \frac{1}{2T}) \delta_{n_0} \rangle  |^2 ~dE
= \int_0^\infty e^{- \frac{t}{T}} | \langle \delta_n, e^{-it H} \delta_{n_0} \rangle |^2 ~dt .
\eeq
Returning to the definition of $\mm_{\varphi}(T)$ in \eqref{defmm}, we use the resolution of the identity $1= \sum_{n \in \ZZ^d} P_n$, where $P_n f(k) = \delta_{nk} f(k)$, and obtain 
\beq\label{eq:mm2}
\mm_{\varphi;n_0}(T) = \frac{1}{T} \int_0^\infty \sum_{n \in \ZZ^d} \varphi(n) 
 | \langle \delta_n, e^{-it( H - i \frac{1}{2T})} \delta_{n_0} \rangle |^2 ~dt .
\eeq
The identity \eqref{eq:planch1} allows us to conclude that 
\beq\label{eq:mm_final1}
\mm_{\varphi;n_0}(T) = \frac{1}{2 \pi T}  \int_{-\infty}^\infty
 \sum_{n \in \ZZ^d} \varphi (n)  |\langle \delta_n, R_H( E + i \frac{1}{2T}) \delta_{n_0} \rangle  |^2 ~dE,
\eeq
proving the identity \eqref{eq:mm0}.
\end{proof}`



For later reference, we state the well known Combes-Thomas estimate in a form that we are going to use frequently in this paper.

\begin{thm}[Combes-Thomas estimate]\label{thm:CombesT} 
Suppose $z\notin\sigma(H) $ and set $\delta:=\dist\left(z,\sigma(H)\right) $.
Then for any $\alpha>0 $ there is a constant $c>0$ such that for all $\delta\leq\alpha $
and all $n,m\in\ZZ^{d} $
   \begin{align}\label{eq:ct1}
      |(H-z)^{-1}(n,m)|~\leq~\frac{2}{\delta}\,e^{-c\delta |n-m|}
   \end{align}
\end{thm}
A proof can be found in \cite{Aizenman} (see also \cite{AizenmanW}). Those works consider operators $H=H_{0}+V$ with rather general $H_{0}$. For the special case of the discrete Schr\"{o}dinger operator, a proof can be found in the review article \cite{Invitation} that follows the idea of \cite{Aizenman}.

\begin{rem}
For $\delta\leq 1 $ one may choose $c=\frac{1}{12d}$ (see \cite{Invitation}). The proof given there allows $c=\frac{1}{12d\,\alpha}$ for $\alpha\geq 1$.
\end{rem}

As a consequence of our proofs, we also get the following estimate on generalized eigenfunctions of the Schr\"{o}dinger operator \emph{outside} the spectrum.
We recall that for any $n_0 \in \ZZ^d$, we define $\overline{\partial}\Lambda_{L}(n_{0}) := \partial \Lambda_L(n_0) \cup {\partial} \Lambda_{L+1} (n_0)$.

\begin{thm}\label{thm:eigenoutside}
If $\psi $ is a generalized eigenfunction of $H$ with generalized eigenvalue $E\notin\sigma(H)$,  
then for any $L$ we have    
\begin{align}\label{eq:bound1}
   \sum_{n\in \overline{\partial}\Lambda_{L}(n_{0})}
   |\psi(n)|~\geq~\frac{\delta}{2 d}\,e^{c\delta L}\;|\psi(n_{0})|
\end{align}
where $\delta=\dist\left(E,\sigma(H)\right) $ and $c$ is as in Theorem \ref{thm:CombesT}\,.
\end{thm}

There is an immediate corollary to the above result showing that generalized solutions to the \Schr equation must grow exponentially in some directions. 

\begin{cor}\label{cor:eigenoutside}
If $\psi $ is a generalized eigenfunction of $H$ with generalized eigenvalue $E\notin\sigma(H)$,  then there are constants $c_{1},c_{2}>0$ and a sequence $ n_{j}\to\infty$ such that
\begin{align}
   |\psi(n_{j})|~\geq~c_{1}\,e^{c_{2}|n_{j}|}\qquad\text{for all } j\in\NN
\end{align}
The constants $c_{1},c_{2}$ depend only on $\dist(E,\sigma(H) $ and the $n_{j}$ can be chosen such that $|n_{j}-j|\leq 1$.
\end{cor}


\section{Estimates of the Green's function}\label{sec:green1}

In this section, we derive basic estimates on the Green's function of a lattice \Schr operator. We begin with a lemma that is an analogue to the Leibnitz rule in a form which will be convenient for this work. The proof is an elementary calculation.

\begin{lemma}\label{lem:prodrule} For any functions $f,u:\ZZ^{d}\to\CC$
\begin{align}
   H(fu)(n)~=~ f(n)\,Hu(n)~+~\sum_{|j|=1}\big(f(n+j)-f(n)\big)\,u(n+j)
   \end{align}
\end{lemma}


The following Proposition \ref{prop:basic} is inspired by Martinelli-Scoppola \cite{MartinelliS}.
We denote by $\chi_{L} $ the indicator function of the box $\Lambda_{L} $.

\begin{prop}\label{prop:basic}
Suppose $\psi:\ZZ^{d}\to \CC $ is a generalized eigenfunction of $H$ with generalized eigenvalue $E$. 
If $z\notin\sigma(H)$, then for any $n\in\Lambda_{L} $
   \bea\label{eq:ef_formula1}
      \psi(n) &=  &(E-z)\,(H-z)^{-1}(\chi_{L}\psi)(n) + ((H-z)^{-1} [H, \chi_L] \psi )(n) \nonumber \\
       & = & (E-z)\,(H-z)^{-1}(\chi_{L}\psi)(n) + ((H-z)^{-1} R_{L, \psi})(n) , 
   \eea
   where the remainder term $R_{L, \psi }(m)$, $m \in \ZZ^d$, is given by 
      \bea\label{eq:rl1}
      R_{L,\psi}(m)  &=: &  ([H, \chi_L] \psi )(m) \nonumber \\ 
       &=&  \sum_{k:|k|=1} [ \chi_{L}(m+k)-\chi_{L}(m) ]\,\psi(m+k) .
       \eea
 The remainder term $R_{L, \psi}$ is supported in the annular region $\overline{\partial}\Lambda_L :={\partial} \Lambda_L \cup {\partial \Lambda}_{L+1}$.  Moreover, we have the bound 
 \begin{align}
     \sum_{n\in\ZZ^{d}}|R_{L,\psi}(n)|~\leq~d\,\sum_{n\in\overline{\partial}\Lambda_{L}}|\psi(n)|
 \end{align}

\end{prop}
   
\begin{proof}
The proof of expression \eqref{eq:ef_formula1} follows from a simple computation using the form of $H$ and Lemma \ref{lem:prodrule}. 
According to Lemma \ref{lem:prodrule} we have
\begin{align*}
   &(H-z)(\chi_{L}\psi)(m)\\~=~&\chi_{L}(m)(H-z)\psi(m)\,
   +\,\sum_{|k|=1}\big(\chi_{L}(m+k)-\chi_{L}(m)\big)\psi(m+k)\\
   =~&\chi_{L}(m)(E-z)\psi(m)\,+\,\sum_{|k|=1}\big(\chi_{L}(m+k)-\chi_{L}(m)\big)\psi(m+k)
\end{align*}
Using $z\notin\sigma(H)$, we get
\begin{align*}
   \chi_{L}\psi~=~(E-z)\big(H-z\big)^{-1}(\chi_{\Lambda}\psi)\,+\,\big(H-z\big)^{-1}R_{L,\psi},
\end{align*}
establishing \eqref{eq:ef_formula1}.
Analyzing $R_{\Lambda,\psi} $ we observe that $\chi_{L}(m+k)-\chi_{L}(m)$ vanishes for all $k$ with $|k|=1 $ unless $m\in\partial\Lambda_{L}\cup\partial\Lambda_{L+1}$. Hence $R_{L,\psi}(m)=0$ if $m\notin\overline{\partial}\Lambda_{L}$. 
Finally we estimate 
\begin{align}
   \sum_{n\in\ZZ^{d}}|R_{L,\psi}(n)|~&\leq~\sum_{n\in\partial\Lambda_{L}}\sum_{{|k|=1}\atop{|n+k|=L+1}} |\psi(n+k)|+\sum_{n\in\partial\Lambda_{L+1}}\sum_{{|k|=1}\atop{|n+k|=L}} |\psi(n+k)|\\
&\leq~d\;\sum_{n\in\overline{\partial}\Lambda_{L}}|\psi(n)|
\end{align}
In the last step we used that any $n\in\Lambda_{L}$ is adjacent to at most $d$ sites $m\in\Lambda_{L+1}$.
\end{proof}

\begin{proof}[Proof (Theorem \ref{thm:eigenoutside})]   
Without loss of generality we may assume $n_{0}=0$. If $E\notin\sigma(H)$, then Proposition \ref{prop:basic} and Theorem \ref{thm:CombesT} give:
\begin{align}
   |\psi(0)|~&\leq~\sum_{m\in \overline{\partial}\Lambda_{L}} |(H-E)^{-1}(0,m)\,R_{L, \psi}(m)|\\
&\leq~\sum_{m\in \overline{\partial}\Lambda_{L}} \left( \frac{2}{\delta} \right) \,
e^{-c \delta |m|}\; |R_{L,\psi}(m)|  \nonumber \\
& \leq \left( \frac{2}{\delta} \right) e^{-c \delta L} \left(  \sum_{m\in \ZZ^{d}} |R_{L, \psi}(m)| \right)  \nonumber \\
&\leq \left(\frac{2 d }{\delta} \right)\,e^{-c \delta L}\; \left( \sum_{k\in \overline{\partial}\Lambda_L} |\psi(k)| \right) ,
\end{align}
This establishes the bound \eqref{eq:bound1}. 
\end{proof}

We turn to Theorem \ref{thm:main2}. This result follows from the following resolvent estimate via Theorem \ref{thm:timeres}.

\begin{thm}\label{thm:mainres}
Let $\psi_E$ be a generalized eigenfunction of $H$ with generalized eigenvalue $E\in\RR$, and suppose that  
 $\psi_E(n_{0})\not=0 $.
We also suppose that $\psi_E$ satisfies a growth condition with respect to a function $\varphi$ satisfying \eqref{assphi}: 
   \begin{align}\label{phibehav2}
      \sum_{n\in\Lambda_{L}(n_{0})}\frac{|\psi_{E}(n)|^{2}}{\varphi(n)}~\leq~A_E   
      \, L^{\nu} \, ,
   \end{align}
   for some $\nu<1 $, a finite constant $A_E > 0$, and for all $L \gg 0$. Then for any $\alpha>1$, 
   there is a finite constant $C_E > 0$ such that
      \begin{align}\label{eq:green2}
      \varepsilon\,\sum_{n\in\ZZ^{d}}\varphi(n)\,\big|(H-E-i\varepsilon)^{-1}(n_{0},n)\big|^{2}~\geq~C_E\,\varepsilon^{\,\alpha\nu-1}\,   ,
   \end{align}
   holds for all $\epsilon > 0$. The finite positive constant $C_E := \frac{C_d |\psi_E(n_0)|}{A_E}$ depends on the dimension $d$ and the energy $E$.
\end{thm}

\begin{proof}
Without loss of generality, we assume $n_{0}=0$.
We set $z=E+i\varepsilon$,  and write $G_z(n,m)$ for the kernel of the resolvent:  $G_{z}(n,m):=(H-z)^{-1}(n,m)$.
By Proposition \ref{prop:basic}, we obtain:
\bea\label{proof1}
   |\psi_E(0)| & \leq  & \varepsilon \sum_{n\in\Lambda_{L}}\,|G_{z}(0,n)|\,    |\psi_E(n)|     +   \sum_{n\in B_L}
    |G_{z}(0,n)|\,|R_{L;\psi_E}(n)|  \nonumber \\
     & :=  &  Q(L) + S(L) .
\eea

First, we consider the second sum $S(L)$ in \eqref{proof1}. From the definition of $R_{L;\psi_E}$ \eqref{eq:rl1},the Combes-Thomas estimate \eqref{eq:ct1} for $G_z(0,n)$,  the bound \eqref{phibehav2} on the generalized eigenfunction $\psi_E$, and the subexponential growth bound on $\varphi$, we obtain the bound
\begin{align} \label{eq:s_L1}
   S(L)~:&=~\sum_{n\in\overline{\partial}\Lambda_{L}}
    |G_{z}(0,n)|\,|R_{L;\psi}(n)|   \nonumber \\
&\leq~\sum_{n\in\overline{\partial}\Lambda_{L}} \left( \frac{2}{\varepsilon} \right) e^{-c\varepsilon|n|} \varphi(n)^{\frac{1}{2}} \left( \frac{|\psi_E(n)|}{\varphi(n)^{\frac{1}{2}}} \right)   \nonumber  \\
&\leq~\frac{2C}{\varepsilon} \left( \sum_{n\in\overline{\partial}\Lambda_{L}} e^{-2c\varepsilon|n|} \varphi (n) \right)^{\frac{1}{2}}  \left( \sum_{n\in\overline{\partial}\Lambda_{L}}\frac{|\psi_E(n)|^2}{\varphi(n)} \right)^{\frac{1}{2}}  \nonumber \\
 & \leq~\frac{2C}{\varepsilon} A_E^{\frac{1}{2}} L^{\frac{\nu}{2}}  e^{-\frac{c\varepsilon}{2} L} \left( \sum_{n\in \overline{\partial}\Lambda_{L}}  e^{-c\varepsilon|n|} \varphi(n)  \right)^{\frac{1}{2}} . 
\end{align}
Now, if we choose $L=\varepsilon^{-\alpha}$ for some $\alpha>1 $, then $S(\varepsilon^{- \alpha})$ goes to zero as $\varepsilon$
goes to zero. So, if we take $\varepsilon$ small enough we may suppose $S(\varepsilon^{-\alpha})\leq \ov{2}|\psi_E(0)|$.

Next, we turn to the first sum $Q(L)$ in \eqref{proof1} with $L = \varepsilon^{- \alpha}$. For $\varepsilon$ sufficiently small, we have 
\begin{align}
 \frac{1}{2} | \psi_E(0)| \leq   Q(\varepsilon^{-\alpha})~:&=~\varepsilon \sum_{|n|\leq \varepsilon^{-\alpha}}\,|G_{z}(0,n)|\,|\psi_E(n)| \nonumber \\
&\leq~\varepsilon \sum_{|n|\leq \varepsilon^{-\alpha}}\,\varphi(n)^{\ov{2}}|G_{z}(0,n)|\,\frac{|\psi_E(n)|}{\varphi(n)^{\ov{2}}}  \nonumber \\
&\leq~\varepsilon^{\ov{2}}\;\Big(\varepsilon\sum_{n\in\ZZ^{d}}\varphi(n)|G_{z}(0,n)|^{2} \Big)^{\ov{2}}\;
\Big(\sum_{|n|\leq \varepsilon^{-\alpha}}\frac{|\psi_E(n)|^{2}}{\varphi(n)}\Big)^{\ov{2}}   \label{eq:s_L2} \\
&\leq~\varepsilon^{\ov{2}}\;\Big(\varepsilon\sum_{n\in\ZZ^{d}}\varphi(n)|G_{z}(0,n)|^{2} \Big)^{\ov{2}}\;A_E^{\ov{2}}\,\varepsilon^{-\ov{2}\alpha\nu}\,,  \nonumber
\end{align}
by assumption \eqref{phibehav2} on the growth of $\psi_E$. It follows that for $\varepsilon$ small enough
\begin{align} \label{eq:s_L3}
  \varepsilon ~ \sum_{n\in\ZZ^{d}}\varphi(n)|G_{E+i\varepsilon}(0,n)|^{2}~\geq~\varepsilon^{\,\alpha\nu-1}\,\frac{|\psi_E(0)|^{2}}{4A_E} ,
\end{align}
establishing the bound \eqref{eq:green2}.
\end{proof}

There are two apparent difficulties with the hypotheses of Theorem \ref{thm:mainres}: 1) the dependence of $A_E$ on $E \in I$, and 2) the requirement that $\psi_E(n_0) \neq 0$ for almost every $E \in I$. In the next two lemmas, we show that these two constraints can be easily dealt with.

We need to be able to choose the constant $A_E$ uniform in $E \in I$. This can be done if we slightly decrease the set $I$.

\begin{lemma}\label{lemma:uniform1}
Let $F: \RR \times \NN \rightarrow \CC$ and suppose that $I \subset \RR$ is a set of positive Lebesgue measure.  Further, suppose there is a positive constant $\theta > 0$, so that for each $E \in I$, there is a constant $A_E > 0$ such that 
\beq
|F(E,n)| \leq A_E n^\theta , \forall n \in \NN .
\eeq
Then, there exists a set $I_0 \subset I$, of positive Lebesgue measure, and a finite constant $A > 0$, independent of $E$, such that for all $E \in I_0$, and all $n \in \NN$, 
\beq
|F(E,n)| \leq A n^\theta .
\eeq
\end{lemma}

\begin{proof}
We set $I_M := \{  E \in I ~|~ | F(E, n) | \leq M n^\theta,  \forall n \in \NN \}$. We then have $I_M \subset I_{M+1}$, and $\bigcup_{M=1}^\infty I_M = I$. 
Thus, $|I_M| \rightarrow |I|$ which implies $|I_M| > 0$,  for all $M$ large enough. 
\end{proof}

\begin{lemma}\label{lemma:non_zero1}
Let $I$ denote an interval as in Theorem \ref{thm:main2} and let $\psi_E$ be  corresponding pb-generalized eigenfunctions for almost every $E \in I$. Then, for some $n_0 \in \ZZ^d$, there exists a subset $I_{n_0} \subset I$ with $|I_{n_0}| > 0$ so that $\psi_E(n_0) \neq 0$, for almost every $E \in I_{n_0}$.
\end{lemma}

\begin{proof}. For each $k \in \ZZ^d$, define ${I}_k := \{ E \in I ~|~ \psi_E(k) \neq 0 \} \subset I$. We have $I = \bigcup_{k \in \ZZ^d} {I}_k$. Since $|I | > 0$, there must exist at least one $n_0 \in \ZZ^d$ such that $|{I}_{n_0}| > 0$.
\end{proof}

\begin{proof}(of Theorem \ref{thm:main2})
We assume the hypotheses of Theorem \ref{thm:main2}: There is a set $I \subset \RR$ with $|I| > 0$,  and $\varphi$ is a weight function satisfying \eqref{assphi}. 
For Lebesgue a.e. $E \in I$, there exists a finite constant $A_E > 0$, a constant $0 \leq \nu <1$,  and a generalized eigenfunction $\psi_{E} $, such that $\psi_{E}(k)\not=0 $, for some $k \in \ZZ^d$, and $\psi_E$ that satisfies
   \begin{align}\label{phibehav3}
      \sum_{|n|\leq L}\frac{|\psi_{E}(n)|^{2}}{\varphi(n)}~\leq~A_E \,L^{\nu}\,,
   \end{align}
 for all $L>0$. By Lemma \ref{lemma:uniform1}, there is a measurable subset $I_{0} \subset I$, with $|I_0| > 0$ and a constant $A_{I_0}$ so that \eqref{phibehav3} holds  
with $A_{I_0}$ replacing $A_E$ for a.e. $E \in I_0$. Furthermore, if we apply Lemma \ref{lemma:non_zero1} to $\{ \psi_E ~|~ E \in I_0 \}$, there is a point $n_0 \in \ZZ^d$, and a subset 
 ${I}_{n_0} \subset I_0$, with 
 $| {I}_{n_0}| > 0$, so that for almost every $E \in I_{n_0}$, the pb-generalized eigenfunctions satisfy $\psi_E (n_0) \neq 0$. From \eqref{eq:s_L3} of 
  Theorem \ref{thm:mainres}, we find that upon taking  $\varepsilon = \frac{1}{2T}$, the resolvent is bounded as
\beq\label{eq:green3}
 \frac{1}{2T} ~ \sum_{n\in\ZZ^{d}}\varphi(n)|G_{E+i \frac{1}{2T}}(n_0,n)|^{2}~\geq~ {(2T)}^{1- \alpha\nu}\,\frac{|\psi_E(n_0)|^{2}}{4A_{I_0}} , 
\eeq
for a.\ e.\ $E \in {I}_{n_0}$.
 We use this bound \eqref{eq:green2}  for the integrand in the representation \eqref{eq:mm0} to obtain
 \bea\label{eq:timeres2}  
  \mm_{\varphi;n_0}(T) & = &\frac{1}{\pi }\,\int_{-\infty}^{\infty}\,\sum_{n\in\ZZ^{d}}\,\varphi(n)\,\Big|\big(H-E-\frac{1}{2T}i\big)^{-1}(n_0,n)\Big|^{2}\,dE  \nonumber \\
   & \geq & \frac{T^{1-\alpha\nu}}{2^{1+ \alpha \nu}  A_{I_0} \pi }\,\int_{{I}_{n_0}}\, | \psi_E(n_0) |^2 ~dE .
 \eea
 The integrand in \eqref{eq:timeres2} is strictly positive, so for any $\alpha>1$,  there is a finite constant $C> 0$ so that 
   \begin{align}
      \mm_{\varphi;n_0}(T)~\geq~C\,T^{1-\alpha\nu}\,   ,
   \end{align}   
   proving the theorem. 
   \end{proof}
   
%

The following corollary follows immediately from the above theorem and its proof.
\begin{cor}
   Suppose $H\psi=E\psi$ with $\psi\in\ell^{2}(\ZZ^d)$ and $\| \psi \|_{2}=1$, then
\begin{enumerate}
   \item 
\begin{align}
   \sum_{n\in\ZZ^{d}} \varepsilon^{2}\, |G_{E+i\varepsilon}(n_{0},n)|^{2}~\geq~|\psi(n_{0})|^{2}
\end{align}
\item If $\psi(n)\leq C e^{-\gamma|n-n_{0}|} $ then
\begin{align}
  \varepsilon \sum_{n\in\ZZ^{d}} \varphi(n) |G_{E+i\varepsilon}(n_{0},n)|^{2}~\geq~C'\,\varepsilon^{-1}
\end{align}
for any $\varphi $ satisfying \eqref{assphi}.
\end{enumerate}
\end{cor}


\section{Application: $\Gamma$-trimmed \Schr operators}\label{sec:trimmed}

In this section we consider our main example, a class of $\Gamma$-trimmed \Schr operators. We prove that for a class of periodic sublattices $\Gamma \subset \ZZ^d$, we can construct sufficiently many bounded generalized eigenfunctions so that, applying the above results, we can prove that these models exhibit nontrivial transport.

We remark that these results extend those of \cite{KiKr2} who studied quantum waveguides. The model for a quantum waveguide consists of compactifying the variables in the $d_1$-directions using periodic boundary conditions. For these $\Gamma$-trimmed models, the  authors proved in Theorem 3.1 that the spectrum is pure point spectrum in any interval $I \subset \{ E ~|~ \dist( E , \sigma(H_{0, \Gamma^c}) > \gamma > 0 \}$, for sufficiently large disorder, where $H_{0, \Gamma^c}$ is the Laplacian restricted to the complement of $\Gamma$ in $\ZZ^d$. They also prove, under some hypotheses, that this is a nontrivial statement in that $[ \inf \Sigma, \inf \Sigma + \epsilon ] \cap  \sigma(H_{0, \Gamma^c}) = \emptyset$, and similarly for  a region near $\sup \Sigma$, for some $\epsilon > 0$. This established localization near the band edges. In addition to this, it is proved in Theorem 3.9 that for a certain family of waveguides, there is nontrivial absolutely continuous spectrum. 
The models discussed in this section include the waveguide models. Although we do not prove the existence of absolutely continuous spectrum, we prove nontrivial transport, and hence dynamical delocalization for these models.


\subsection{Basic definitions and the model}\label{subsec:defn_trim1}

We begin with the definition of a $\Gamma$-trimmed \Schr operator on $\ZZ^d$. 

\begin{definition}
Let $\Gamma $ be a nontrivial subset of $\ZZ^{d}$. A Hamiltonian $H=H_{0}+V$ on $\ell^{2}(\ZZ^{d}) $ is called $\Gamma$-trimmed if
$V(n)=0 $ whenever $n\not\in\Gamma$.
If the nonzero potential is given by independent, identically distributed random variables $V_\omega(n)$, that is
\begin{align}
\{ V(n) \}_{n \in\Gamma}~=~\{ V_{\omega}(n) \}_{n \in\Gamma}
\end{align}
we call it a \emph{$\Gamma $-trimmed Anderson model}. In this case, we always assume that the common distribution of the $V_{\omega}(n)$ has a bounded density.
\end{definition}

Our fundamental assumption on $\Gamma$ is the following.
We are interested in cases for which $\Gamma$ is a periodic subset of $\ZZ^d$ and of the following special form. In order to describe this, we introduce the following subspaces. 

\begin{definition} For $p\in\NN$, and for each $1\leq i\leq m$, 
   we set
\begin{align}\ZZ_{i,p}^{m}~&:=~\{ (n_{1},\ldots,n_{m})\mid n_{i}\in p\,\ZZ, ~n_j \in \ZZ, j \neq i \} \subset \ZZ^m, \\
\intertext{and for $\rho=(\rho_{1},\ldots,\rho_{m})\in\NN^{m}$}
   W_{\rho}^{m}~&:=~\bigcup_{i=1}^{m}\,\ZZ_{i,\rho_{i}}^{m}  \subset \ZZ^m \,.
\end{align}
\end{definition}

\begin{assum}\label{assupm:trim2}
The periodic set $\Gamma\subset\ZZ^{d} $ is of the form
\begin{align}\label{sublat}
   \Gamma~=~W_{\rho}^{d_{1}}\times\ZZ^{d_{2}}
\end{align}
where $\rho=(\rho_{1},\ldots,\rho_{d_{1}})$ with $\rho_{i}\geq 2$ and $d=d_{1}+d_{2}$.
\end{assum}

The subspace $ W_{\rho}^{d_{1}}$ forms a $d_{1}$-dimensional grid in $\ZZ^{d} $. In particular if $d_{2}=0$ so that $d_1 = d$, we learn from the papers cited above that there is a Wegner estimate for all energies outside a discrete set for the 
$\Gamma $-trimmed Anderson mode. In this case, the \Schr operator $H$ has pure point spectrum for high disorder. In the other extreme, $d_{1}=0$, the $\Gamma $-trimmed Anderson model is just the standard Anderson model on $\ZZ^d$.

We are interested in families of generalized eigenfunctions for the $\Gamma$-trimmed \Schr operator. For this, we introduce the following notation.
\begin{definition}
   If $\Gamma $ is as in \eqref{sublat} we set
\begin{align}
   e(p_{1},\ldots,p_{d_{1}})~&=~2\sum_{i=1}^{d_{1}} \cos{(\frac{\pi}{2\rho_{i}}p_{i})}\\
  \text{and}\qquad E(p_{1},\ldots,p_{d})~&=~2\sum_{i=1}^{d_{1}} \cos{(\frac{\pi}{2\rho_{i}}p_{i})}~+~2\sum_{j=1}^{d_{2}}\cos{(\frac{\pi}{2}p_{d_{1}+j})}\,.
\end{align}
We also set
\begin{align}
   \Ss_{0}~&=~\{\, e(k_{1},\ldots k_{d_{1}})\mid k_{i}\in\ZZ\setminus\{ 0 \}\}\\
\text{and} \qquad  \Ss~&=~\{\, E(k_{1},\ldots k_{d_{1}},\kappa_{1},\ldots,\kappa_{d_{2}})\mid k_{i}\in\ZZ\setminus\{ 0 \}, \kappa_{j}\in\RR\}
\end{align}
\end{definition}

\begin{prop}
Suppose $\Gamma \subset \ZZ^d$ satisfies Assumption \ref{assupm:trim2}, and that the \Schr operator $H$ is $\Gamma$-trimmed.
Then the set
   \begin{align}
      \Ss~=~\Ss_{0}\,+[-2d_{2},2d_{2}]
   \end{align}
 consists of generalized eigenvalues of $H$ and is hence contained in $\sigma(H) $.
The functions
   \begin{align}\label{eq:gen_ef1}
      \psi(n_{1},\ldots,n_{d})~=~\prod_{i=1}^{d_{1}}\,\sin{(\frac{\pi}{2\rho_{i}}k_{i}n_{i})}~\times~
      \prod_{j=1}^{d_{2}}\,\e^{\i\frac{\pi}{2}\kappa_{j}n_{d_{1}+j}}
   \end{align}
are a set of bounded, generalized eigenfunctions of $H$ corresponding to generalized eigenvalues $E(k_{1},\ldots k_{d_{1}},\kappa_{1},\ldots,\kappa_{d_{2}}) $.
\end{prop}

\begin{rem}\quad
\begin{enumerate}
   \item The set $\Ss $ has no isolated points ($d_{2}\geq 1$) hence $\Ss\subset\sigma_{ess}(H) $.
   \item We make \emph{no} assumption on the behavior of $V$ on the set $\Gamma$ except that it be bounded. It may be random or even zero on $\Gamma $.
   \item It is an immediate consequence of Corollary \ref{cor:bnd} that $\mm_{q}$ is unbounded if $q>d-1 $. If $d_{2}\geq 2$ we can do better.
\end{enumerate}
\end{rem}

Our main result on these $\Gamma$-trimmed models is that the model exhibits nontrivial transport for any bounded potential supported on $\Gamma$.

\begin{thm}\label{thm:trimmed+}
Suppose $\Gamma $ is as in \eqref{sublat} with $d_{2}\geq 2$. If $H$ is $\Gamma $-trimmed, then for some $n_0 \in \Gamma^c \subset \ZZ^d$, 
we have for $q>d_{1}+1$
   \begin{align}
      \mm_{q;n_0}(T)~\geq~C\,T\,,
   \end{align}
   and for $d_{1}< q< d_{1}+1 $
   \begin{align}
      \mm_{q;n_0}(T)~\geq~C\,T^{q-d_{1}-\varepsilon}\,,
   \end{align}
\end{thm}

\begin{rem}\label{rmk:kl1}
We compare this result with the results obtained by applying a discrete version of Theorem \ref{thm:kl1} (\cite{KL}) to $\Gamma$-trimmed \Schr operators. 
Since the generalized eigenfunctions $\psi$ in \eqref{eq:gen_ef1} are products of bounded generalized eigenfunctions on $\ZZ^{d_1}$ with functions in $\ell^2(\ZZ^{d_2})$, the term in \eqref{eq:kl1} can be bounded by $R^{d_1-\gamma}$. If we could prove that the spectral measure $\mu_\psi^H$ is $\alpha$-continuous, then we would obtain the lower bound \eqref{eq:kl2} with $\gamma = d_1$. Unfortunately, we do not know how to prove the $\alpha$-continuity of the spectral measure. 
\end{rem}

Furthermore, if we restrict to $\Gamma$-trimmed Anderson models, we have the following theorem.  

\begin{thm}
The $\Gamma$-trimmed Anderson model at high disorder exhibits dynamical localization on the complement of the spectrum of $H_{\Gamma^c}$ in the deterministic spectrum $\Sigma$. Furthermore, the transport moments are bounded from below, so the model also exhibits dynamical delocalization. There exist $n_0 \in \Gamma^c \subset \ZZ^d$, so that for $q>d_{1}+1 $, we have
   \begin{align}
      \mm_{q;n_0}(T)~\geq~C\,T\,,
   \end{align}
   and for $d_{1}< q< d_{1}+1 $
   \begin{align}
      \mm_{q;n_0}(T)~\geq~C\,T^{q-d_{1}-\varepsilon}\,,
   \end{align}
\end{thm}

The first part of the theorem on localization is proved in \cite{KiKr1}. For the second part, we use the fact that for each $E \in \mathcal{S}$, there is a bounded generalized eigenfunction. The proof is presented in section \ref{subsec:gamma_pfs1}


\subsection{Proofs for the $\Gamma$-trimmed models}\label{subsec:gamma_pfs1}

We turn to the proof of Theorem \ref{thm:trimmed+}. Given Theorem \ref{thm:main2} Theorem \ref{thm:trimmed+} follows from the following Proposition.

\begin{prop}\label{prop:trimmed}
   Suppose $\Gamma $ is as in \eqref{sublat} with $d_{2}\geq 2$. If $H$ is $\Gamma $-trimmed and
\begin{align}
   E~\in~\Ss_{0}~+~(-2d_{2},2d_{2})
\end{align}
then there is a solution $\psi $
 of $H\psi=E\psi $ such that for $q>d_{1}+1 $
\begin{align}
    \sum_{|n|\leq L}\frac{|\psi_{E}(n)|^{2}}{\langle n \rangle^{q}}~&\leq~C&&\text{for }q>d_{1}+1    \label{eq:trim_bd1} \\  
\text{and }\quad\sum_{|n|\leq L}\frac{|\psi_{E}(n)|^{2}}{\langle n \rangle^{q}}~&\leq~A\,L^{d_{1}+1-q}&&\text{for }d_{1}<q<d_{1}+1\,. \label{eq:trim_bd2}
\end{align}

\end{prop}

To prove this Proposition we shall use the following Lemma.
\begin{lemma}\label{lem:free}
   Consider the free Hamiltonian $H_{0}$ on $\ell^{2}(\ZZ\times\ZZ^{m})$ with $m\geq 1$. For any $E\in\big(-2(m+1),2(m+1)\big)$ there is a solution $\psi$ of $H_{0}\psi=E\psi$ such that
\begin{align}\label{free}
   \sup_{k\in\ZZ}\,\sum_{n\in\ZZ^{m}}\,|\psi(k,n)|^{2}~<~\infty\,.
\end{align}
\end{lemma}

\begin{proofn}[\,(Proposition \ref{prop:trimmed})]  
   Any $E\in\Ss_{0}+(-2d_{2},2d_{2})$ can be written as
\begin{align}
   E~=~2\sum_{i=1}^{d_{1}} \cos{(\frac{\pi}{2\rho_{i}}k_{i})}~+~e
\end{align}
with $k_{i}\in\ZZ\setminus\{ 0 \}$ and $e\in\left(-2d_{2},2d_{2}\right) $.
Let $h_{0}$ denote the discrete Laplacian on $\ZZ^{d_{2}}$.
Since $d_{2}\geq 2$ by Lemma \ref{lem:free} there is a solution $\psi_{2}:\ZZ\times\ZZ^{d_{2}-1} $ of $h_{0}\psi_{2}=e\psi_{2}$
satisfying
\begin{align}
   \sup_{m_{1}}\sum_{m_{2},\ldots,m_{d_{2}}}\;|\psi_{2}(m_{1},\ldots,m_{d_{2}})|^{2}~<~\infty\,.
\end{align}
Set
\begin{align}
   \psi_{1}(n_{1},\ldots,n_{d_{1}})~=~\prod_{i=1}^{d_{1}}\,\sin{(\frac{\pi}{2\rho_{i}}k_{i}n_{i})}\,,
\end{align}
then
\begin{align}
   \Psi(n_{1},\ldots,n_{d_{1}},m_{1},\ldots,m_{d_{2}})~=~\psi_{1}(n_{1},\ldots,n_{d_{1}})\,\psi_{2}(m_{1},\ldots,m_{d_{2}})
\end{align}
solves $H_{0}\Psi=H\Psi=E\Psi$.
\end{proofn}

\medskip

\noindent
We now present the proof of Lemma \ref{lem:free}.

\medskip

\begin{proofn}[\,(Lemma \ref{lem:free})]
For $E\in\big(-2(m+1),2(m+1)\big)$, there are $\theta_{2},\ldots,\theta_{d_{m+1}}\in(0,2\pi)$ such that
\beq\label{eq:bound10}
\left|   E -  \sum_{j=2}^{m+1} 2 \cos \theta_j \right| < 2 .
\eeq
Let $\mathcal{S}_E \subset \prod_{j=2}^{m+1} \TT^m$ be the set of $(\theta_2, \ldots, \theta_{m+1})$ so that \eqref{eq:bound10} is satisfied. 
We define $\theta_1^+  := \theta_1( \theta_2, \ldots, \theta_{m+1})  \in (0, \pi)$ so that 
\beq
2 \cos \theta_1^+ = E -  \sum_{j=2}^{m+1} 2 \cos \theta_j ,
\eeq
that is,  
\beq\label{eq:inv_cos1}
\theta_1^+ = \cos^{-1} \left( \frac{E}{2} -  \sum_{j=2}^{m+1}  \cos \theta_j \right) ,
\eeq
where $\cos^{-1} : (-1,1) \rightarrow (0, \pi)$, since the argument in \eqref{eq:inv_cos1} has absolute value less than one.
Let ${\mathcal{S}_E}$ be the open ball in $\TT^m$ described above so that the characteristic function $\chi_{\mathcal{S}_E} \in L^2 ( \TT^m)$. 
We construct a a function $\psi_E (k,n)$ by
 \beq
 \Psi_E(k,n) :=  \int_{\TT^m} e^{i k \theta_1^+} e^{i \sum_{j=2}^{m+1} n_j \theta_j} \chi_{\mathcal{S}_E}( \theta_2, \ldots, \theta_{m+1} ) \prod_{\ell=2}^{m+1} ~d \theta_\ell .
 \eeq
 It follows from the fact that the discrete Laplacian $H_0$ acts by multiplication by $\sum_{j=1}^{m+1} 2 \cos \theta_j$, that $H_0 \psi_E = E  \psi_E$, so $\psi_E$ is a generalized eigenfunction.  
The bound \eqref{free} follows from the Plancherel Theorem on $\ell^2( \ZZ^m)$. 
\end{proofn}

\medskip

Given these two preparatory  results, we prove the main theorem, Theorem 
\ref{thm:trimmed+}, on the existence of nontrivial transport for the periodic $\Gamma$-trimmed Anderson model.  

\medskip

\begin{proofn}[\,(Theorem \ref{thm:trimmed+})]  
By Proposition \ref{prop:trimmed}, for $E~\in~\Ss_{0}~+~(-2d_{2},2d_{2})$, there exists a generalized eigenfunction $\psi_E$ of $H$
satisfying the bounds \eqref{eq:trim_bd1} and \eqref{eq:trim_bd2}, depending on the value of $q$. The result follows from Theorem \ref{thm:mainres}.
\end{proofn}


\section{Eigenfunctions and Spectral Type}\label{sec:spectralType}

In this section, we relate the behavior of solutions of the eigenvalue equation of a  discrete Schr\"odinger operator to the spectral type of the operator. As we have seen, self-adjoint \Schr operators
coming from difference operators on $\ell^2({\mathbb Z}^d)$ may have solutions
corresponding to eigenvalue equations associated with a number $E$, which
are not in $\ell^2({\mathbb Z}^d)$. Theorem \ref{thm:mainres} relates the growth properties of such solutions to the transport properties of the operator. In this section, we relate properties of generalized eigenfunctions to the spectral properties of the \Schr operator. 

For example, we know that if the solutions are in
$\ell^2({\mathbb Z}^d)$, then $E$ is in the point spectrum of the 
operator.        
In the case that the solutions are not in $\ell^2({\mathbb Z}^d)$, 
there is a body of literature discussing the relationship between the behavior of such solutions and the spectrum of the operator.  We already mentioned  the results of 
Berezansky \cite{Be} and Sch'nol \cite{Schnol} on polynomially bounded generalized eigenfunctions. 
The results of Deift-Simon \cite{DS},  
Gilbert-Pearson \cite{GP},  Jitomirskaya-Last \cite{JL}, Christ-Kisalev-Last \cite{CKL} on approximate eigenfunctions are a few that study different aspects
of such a relationship. 




%

\begin{thm}\label{thmk1}
Let $H= - \Delta + V$ be a bounded self-adjoint \Schr operator on $\ell^2({\mathbb Z}^d)$.  
Suppose that $\psi$ is a generalized eigenfunction of $H$ with generalized eigenvalue $E$,  and let $n \in \ZZ^d$ be such that $\psi(n) \neq 0$. 
Then, for any $\alpha > 1$, and any $0 < \gamma < 1$, 
\begin{enumerate}
\item The condition on the generalized eigenfunction $\psi$:
$$
\liminf_{\varepsilon \rightarrow 0} \varepsilon^{\gamma} \sum_{|m - n|\leq [\varepsilon^{-\alpha}]} |\psi(m)|^{2} =0,  
$$
implies that 
$$
\limsup_{\varepsilon \rightarrow 0} ~ \varepsilon^{1-\gamma} Im \langle \delta_n, ~ (H - E - i\varepsilon)^{-1} \delta_n\rangle = \infty. 
$$
\item The condition on the generalized eigenfunction $\psi$:
$$
\limsup_{\varepsilon \rightarrow 0} \varepsilon^{\gamma} \sum_{|m - n|\leq [\varepsilon^{-\alpha}]} |\psi(m)|^{2} < \infty,  
$$
implies 
$$
\liminf_{\varepsilon \rightarrow 0} ~ \varepsilon^{1-\gamma} Im \langle \delta_n, ~ (H - E - i\varepsilon)^{-1} \delta_n\rangle > 0. 
$$ 
\end{enumerate}
\end{thm}


\begin{proofn}
We use the inequalities \eqref{eq:s_L1}, \eqref{eq:s_L2}, and \eqref{eq:s_L3}, with $\varphi = 1$, to get the bound,
$$
\Big(\sum_{|n|\leq \varepsilon^{-\alpha}}\frac{|\psi(n)|^{2}}{\varphi(n)}\Big)
\leq  \varepsilon^{1- \gamma }\sum_{|n-m|\leq [\varepsilon^{-\alpha}]} |\psi(m)|^{2}.
$$
Then the inequality \eqref{eq:s_L3} gives, 
\begin{equation}\label{eqnKri20}
|\psi(n)|^2  \leq \varepsilon^\gamma Im \Big \langle \delta_n, (H - E - i\varepsilon )^{-1} \delta_n \Big) ~ 
\Big( \varepsilon^{ 1- \gamma }  \sum_{|n-m|\leq [\varepsilon^{-\alpha}]}|\psi(m)|^{2}\Big).
\end{equation}
The above inequality can be restated as: 
\begin{equation}\label{eqnKri21}
\frac{|\psi(n)|^2}{\varepsilon^{ 1- \gamma  } \displaystyle{ \sum_{|n-m|\leq [\varepsilon^{-\alpha}]} | \psi(m) |^{2}} } \leq \varepsilon^\gamma Im \Big( \langle \delta_n, (H - E - i\varepsilon )^{-1} \delta_n \rangle \Big). 
\end{equation}
Since $\psi(n)$ is not zero, the first part of the theorem follows. For the second part, we 
write 
 \begin{equation}\label{eqnKri22}
\frac{|\psi(n)|^2}{\varepsilon^\gamma Im \Big( \langle \delta_n, (H - E - i\varepsilon )^{-1} \delta_n \rangle \Big)}  \leq  \varepsilon^{ 1- \gamma }  
\displaystyle{ \sum_{|n-m|\leq [\varepsilon^{-\alpha}]} | \psi(m) |^{2}}  .
\end{equation}
\end{proofn}

\medskip
\begin{rem}
We note that in Theorem \ref{thmk1} we start with solutions of eigenvalue equations
of the operator $H$ and do not consider approximate generalized eigenvectors. If we do take approximate eigenvectors to begin with, then we could take $\alpha = \gamma=1$ in the theorem and recover the results of Kiselev-Last \cite{KL}.  Our results should be thought of as lower bounds on the solutions of eigenvalue equations
associated with $H$.  Hence in the following corollary, a spectral type of a
part of the spectrum implies lower bounds on solutions of eigenvalue equations for that part of the spectrum.
\end{rem}
\medskip

We relate the behavior of the boundary-values of the resolvent to the spectral measures.
For any self-adjoint operator $H$ on $\ell^2({\mathbb Z}^d)$, and for any $n \in \ZZ^d$ and $\gamma \in (0,1)$, we define 
the set $S_{\gamma,n}$ as the set of $E$ for which there is a generalized eigenfunction $\varphi_E$ of $H$, with $H \varphi_E = E \varphi_E$, with $\varphi_E(n) \neq 0$, 
whose local $\ell^2$ norm satisfies a certain growth condition:
\begin{eqnarray}\label{eqnk1}
S_{n, \gamma} & :=  & \left\{  E \in \RR ~:~  \exists \varphi_E, ~ \varphi_E(n) \neq 0, H \varphi_E = E \varphi_E,~ and ~\exists ~   \alpha > 1 ,  \right. \nonumber \\
 & & \left. ~ so ~ that ~ \varphi_E ~ satisfies ~ condition ~ \mathcal{G}_E \right\} 
\eea 
where the growth condition $\mathcal{G}_E$ is
 \beq\label{eq:growth1}
\mathcal{G}_E: ~~  \liminf_{\varepsilon \rightarrow 0}  \varepsilon^{1-\gamma}\displaystyle{ \sum_{|n-m|\leq [\varepsilon^{-\alpha}]} |\varphi_E(m)|^2 = 0}   .
 \eeq
 

We denote by $\mu_{k, H}$, the spectral measure of $H$ associated with the vector $\delta_k, ~~ k \in {\mathbb Z}^d$, where $\{ \delta_k ~|~ k \in \ZZ^d \}$ is the standard orthonormal basis. We denote by $\mu_{k, H, \beta}$, the singular part of $\mu_{k, H}$ that is absolutely continuous 
with respect to the Hausdorff measure $h^\beta$, see Last
 \cite{last} or Rogers \cite{Ro} or Demuth-Krishna \cite[Theorm 1.1.5]{DK}.
An immediate corollary of the above theorem is the following result on the spectral measures.

\begin{cor}\label{cor1}
Let $H$ be a self-adjoint bounded \Schr operator on $\ell^2({\mathbb Z}^d)$ as in Theorem \ref{thmk1}.
Let $\gamma \in (0, ~1)$.  Then the ${1-\gamma}$ Hausdorff measure of 
$S_{n, \gamma}$ is zero. This means that $\mu_{n,H,\beta}$, the $\beta$-singular part of $\mu_{n, H}$, is zero for any $\beta > 1- \gamma$. 
\end{cor}


\begin{proofn} This proof is a direct application of a theorem of Rogers-Taylor \cite[Theorem 67]{Ro}.  (See also Last \cite{last}, Demuth-Krishna \cite[Theorem 1.3.6, Theorem 1.1.9]{DK}). 
\end{proofn}


\end{document}